\documentclass[sigconf, nonacm]{acmart}
\usepackage[utf8]{inputenc}
\usepackage{amsmath}
\usepackage{balance}
\usepackage{fancyhdr}
\makeatletter
\makeatletter
\newcommand*\mysize{%
  \@setfontsize\mysize{10.0}{11.0}%
}
\makeatother\makeatother

\usepackage{amssymb}

\usepackage{amsthm}
\usepackage{tikz}
\usepackage{graphicx}
\usepackage{float}
\usepackage{algorithm}
\usepackage[noend]{algpseudocode}
\usepackage{tabularx}
\usepackage{arydshln}
\usepackage{multirow}
\usepackage{wrapfig}
\usepackage{stdclsdv}
\usepackage{comment}
\usepackage{xspace}
\usepackage{graphics}
\usepackage{subcaption}
\usepackage{hyperref}
\usepackage{caption}
\usepackage{lipsum}
\usepackage{bbm} % for indicator functions
\usepackage{afterpage,natbib,lipsum}
\usepackage{lipsum}
\def\etal.{et\penalty50\ al.}
\theoremstyle{plain}
\newtheorem{theorem}{Theorem}[section]
\newtheorem{lemma}[theorem]{Lemma}
\newtheorem{proposition}[theorem]{Proposition}
\newtheorem{corollary}[theorem]{Corollary}

\newtheorem{claim}[theorem]{Claim}
\theoremstyle{definition}
\newtheorem{definition}{Definition}[section]
\newtheorem{example}[definition]{Example}

\theoremstyle{remark}

\theoremstyle{plain}
\newtheorem*{theorem*}{Theorem}

\DeclareMathOperator{\OPT}{OPT}
\DeclareMathOperator{\ALG}{ALG}
\DeclareMathOperator{\ADV}{ADV}  % adversary

\DeclareMathOperator{\util}{util}           % utilization

\newcommand{\norm}[1]{\left\lVert#1\right\rVert}

\newcommand{\LF}{\textit{LastFit}\xspace}
\newcommand{\FF}{\textit{FirstFit}\xspace}
\newcommand{\NF}{\textit{NextFit}\xspace}
\newcommand{\BF}{\textit{BestFit}\xspace}
\newcommand{\AF}{\textit{AnyFit}\xspace}
\newcommand{\RF}{\textit{RandomFit}\xspace}
\newcommand{\MTF}{\textit{MTF}\xspace}
\newcommand{\HA}{\textit{HA}\xspace}
\newcommand{\MFF}{\textit{MFF}\xspace}
\newcommand{\WF}{\textit{WorstFit}\xspace}

\newcommand{\MNF}{\textit{MNF}\xspace}
\newcommand{\Greedy}{\textit{Greedy}\xspace}

\newcommand{\RSiC}{\textit{RSiC}\xspace}
\algnewcommand{\LineComment}[1]{\State \(\triangleright\) #1}

\newcommand{\cA}{\ALG}
\newcommand{\cF}{\mathcal F}
\newcommand{\cR}{\mathcal \sigma}

\DeclareMathOperator{\Ex}{\mathbb{E}}           % expected value 

\newcommand{\vecs}{{s}}
\DeclareMathOperator{\myspan}{{span}}

\title{Renting Servers for Multi-Parameter Jobs in the Cloud}

\author{%
    \texorpdfstring{
        \begin{tabular}{c@{\hskip 2cm}c}
            Yaqiao Li & Mahtab Masoori \\
            Concordia University, CSSE & Concordia University, CSSE \\
            \href{mailto:yaqiao.li@concordia.ca}{yaqiao.li@concordia.ca} & \href{mailto:mahtab.masoori@concordia.ca}{mahtab.masoori@concordia.ca}
        \end{tabular}
        }{Author Name 1 and Author Name 2} % Use \texorpdfstring to handle special characters for PDF bookmarks
    \and
    \texorpdfstring{
        \begin{tabular}{c@{\hskip 2cm}c}
            Lata Narayanan & Denis Pankratov \\
            Concordia University, CSSE & Concordia University, CSSE \\
            \href{mailto:lata.narayanan@concordia.ca}{lata.narayanan@concordia.ca} & \href{mailto:denis.pankratov@concordia.ca}{denis.pankratov@concordia.ca}
        \end{tabular}
        }{Author Name 3 and Author Name 4}
}
\begin{abstract}
    We study the Renting Servers in the Cloud problem (\RSiC) in multiple dimensions. In this problem, a sequence of multi-parameter jobs must be scheduled on servers that can be rented on-demand. Each job has an arrival time, a finishing time, and a multi-dimensional size vector that specifies its resource demands. Each server has a multi-dimensional capacity and jobs can be scheduled on a server as long as in each dimension the sum of sizes of jobs does not exceed the capacity of the server in that dimension. The goal is to minimize the total rental time of servers needed to process the job sequence. 

    \AF algorithms do not rent new servers to accommodate a job unless they have to. We introduce a sub-family of \AF algorithms, which we call monotone \AF algorithms. We show that monotone \AF algorithms have a tight competitive ratio of $\Theta(d \mu)$, where $d$ is the dimension of the problem and $\mu$ is the ratio between the maximum and minimum duration of jobs in the input sequence. We also show that upper bounds for the \RSiC problem obey the direct-sum property with respect to dimension $d$, that is we show how to transform $1$-dimensional algorithms for \RSiC to work in the $d$-dimensional setting with competitive ratio scaling by a factor of $d$. As a corollary, we obtain an $O(d\sqrt{\log \mu})$ upper bound for $d$-dimensional clairvoyant \RSiC.
    We also establish a lower bound of $\widetilde{\Omega}(d \mu)$ for both deterministic and randomized algorithms for $d$-dimensional non-clairvoyant \RSiC, under the assumption that $\mu \le \log d - 2$.

    Lastly, we propose a natural greedy algorithm, which we call \Greedy. This is a clairvoyant algorithm that belongs to the monotone \AF family of algorithms, thus, it has competitive ratio $\Theta(d \mu)$. Our experimental results indicate that \Greedy performs better or as well as all other previously proposed algorithms, for almost all the settings of arrival rates and values of $\mu$ and $d$ that we implemented. 
\end{abstract}
\makeatletter
\def\@copyrightspace{\relax}
\makeatother
\begin{document}

\maketitle
\pagestyle{plain}
\section{Introduction}

One of the most famous problems in online computation is the \emph{bin packing} problem which has received a lot of attention among researchers~\cite{coffman1984approximation, coffman2013bin}. Given a set of items with a positive size, and a bin capacity, the objective is to pack all the items in the minimum possible number of bins such that the total size of items assigned to a bin does not exceed the bin capacity. For simplicity, it is generally assumed that the item sizes lie between $0$ and $1$ and the bins all have unit capacity. This problem is online in the sense that the items come in a sequential manner and the algorithm has to place a new item in a bin without having any knowledge about the upcoming items. 
Dynamic bin packing  is a generalization of the bin packing problem~\cite{coffman1983dynamic}, in which items not only have an arrival time and a size, but also a {\em duration}. 
The objective function is the same as the classic bin packing problem; minimizing the total number of bins that is used to pack all the items. Dynamic bin packing has been extensively used to model various resource allocation problems, highlighting its adaptability in optimizing packing scenarios for dynamic settings~\cite{stolyar2013large, karwayun2018dynamic, jiang2012joint,stolyar2013infinite}.

Li et al.~\cite{li2015dynamic} introduced \emph{MinUsageTime Dynamic Bin Packing}, a new variant of Dynamic Bin Packing. The problem is also known as \emph{Renting Servers in the Cloud (\RSiC)}, and 
is primarily motivated by job allocation to servers in the cloud.  For example, users make requests for virtual machines (VMs) with specific requirements to a cloud service provider such as Microsoft Azure, which then has to assign VMs to physical servers with sufficient capacity. The power and other costs incurred by the service provider are directly proportional to the total duration that servers are kept active. Optimal assignment decisions can reduce fragmentation which can result in dramatic cost savings \cite{hadary2020protean}. 
As another application, cloud gaming companies such as GaiKai, OnLive, and StreamMyGame rent servers from public cloud companies and are charged using a pay-as-you-go model. A customer's request to play a game is assigned to one of the rented servers that has enough capacity to serve the request. The rental cost paid by the gaming company is directly proportional to the duration of time that the servers are rented. 

These situations are modeled by the \RSiC problem, which has been studied in both non-clairvoyant  and clairvoyant settings, see, e.g. ~\cite{ren2018combinatorial,RSiC2015MTF,Clairvoyant_HA}. Recently, Murhekar et al. ~\cite{murhekar2023brief} initiated the study of \RSiC for multi-parameter jobs. 
In this setting, jobs with resource requirements for multiple parameters (e.g. number of GPUs, memory, network bandwidth etc.) arrive to the system in an online manner and must be assigned by the algorithm to servers with fixed capacity along each of these dimensions.   In the non-clairvoyant setting, the arrival time of a job and its resource requirements, given as a {\em size vector}, are revealed to the algorithm when the job arrives, but its duration is only known when it departs. In the clairvoyant setting, the finishing time of a job is also known to the algorithm when it arrives. Jobs  must be assigned to servers immediately after arriving, and the algorithm's decisions are irrevocable, as the cost of moving a job mid-execution to another server is assumed to be prohibitive.  The objective of the \RSiC problem is to minimize the total cost of all rented servers, where  the cost of a server is proportional to the duration for which it is rented/utilized.

The performance of online algorithms is measured by the notion of \emph{competitive ratio}. This ratio represents the worst-case comparison, considering all inputs, between the cost achieved by the online algorithm and the cost achieved by an optimal offline solution that has complete knowledge of the entire input instance in advance.

\subsection*{Our contributions}
In this paper, we study the $d$-dimensional \RSiC problem, when job sizes are $d$-dimensional vectors and the size of a job in any dimension is normalized to lie between 0 and 1, and the server capacity in each dimension is 1. We consider both the clairvoyant and non-clairvoyant versions of the problem. Our main results in this paper are summarized below:

\begin{itemize}

    \item We introduce a subset of \AF algorithms called monotone \AF algorithms and propose a new clairvoyant algorithm within this category called \emph{\Greedy}, which assigns a new job to the server with enough remaining capacity that would incur the {\em least additional cost.} We prove that all monotone \AF algorithms including \Greedy have a competitive ratio of $3d \mu+1$, where $\mu$ is the ratio between the maximum and minimum duration of jobs in the input sequence. The proof uses a new technique compared to those used in \cite{RSiC2015MTF} and \cite{ren2018combinatorial}. 

We remark here that the proof of a $6\mu+8$ upper bound on the competitive ratio of \MTF given in \cite{RSiC2015MTF} also works for \Greedy for 1-dimensional \RSiC.  However, we are able to derive a better upper bound with our technique, which moreover generalizes to any monotone \AF algorithm and $d$-dimensional \RSiC. 
    The upper bound is  tight as a lower bound of $\Omega(d \mu)$ was shown on the competitive ratio of \AF algorithms in \cite{murhekar2023brief}.

% \yaqiao{I suggest to point out the bound $6\mu+8$ in \cite{RSiC2015MTF} for MTF also works for greedy, so there is already an $O(d\mu)$ upper bound. We should emphasize that our main contribution here is a NEW method to improve this bound to $3\mu$. Our method is different from two main methods existing in literature, the one in \cite{RSiC2015MTF}, and the one in other works(which deals with a span part and non-span part separately). However, we do use one key idea from \cite{RSiC2015MTF}, which is to use the order of servers at a specific time of our interest. It turns out that our proof works for monotone algorithms including greedy, MTF, lastFit, etc.}    

    \item 
    %We analyze the clairvoyant problem in the $d$-dimensional setting for the first time. 
 
    We demonstrate a very general direct sum property of the \RSiC problem by showing how to transform any algorithm $\ALG$ for dimension 1 to a corresponding algorithm $\cA^{\oplus d}$ for dimension $d$, with competitive ratio scaling \emph{exactly} by a factor of $d$. As a corollary, we obtain the first clairvoyant algorithm for $d$-dimensional \RSiC, with competitive ratio $\Theta(d\sqrt{\log \mu})$.
%    \yaqiao{I added exactly after scaling, and added the general direct sum property wording.}

    \item We adapt, for the first time, an online graph coloring lower bound construction in \cite{HScoloring} to prove lower bounds for $d$-dimensional \RSiC. Specifically, we prove a lower bound of $\widetilde{\Omega}(d \mu)$ for both deterministic and randomized algorithms for  $d$-dimensional non-clairvoyant \RSiC when $\mu \le \log d - 2$, and a lower bound of  $\widetilde{\Omega}(d)$ for the clairvoyant case. Connections between online vector bin packing and online graph coloring have been explored in ~\cite{azar2013tightVBP,LI202370VBP}. 
 %   \yaqiao{I edited a little bit to emphasize our contribution.}

     \item We conduct experiments in the average-case scenario, evaluating nearly all existing algorithms for  \RSiC  using randomly generated synthetic data. Our findings show that the \Greedy algorithm  outperforms other algorithms, whether clairvoyant or non-clairvoyant, in the vast majority of cases. 
\end{itemize}

\subsection*{Organization} 
 The related work is summarized in Section~\ref{sec:Previous-Work}. Formal definitions and notation are introduced in Section~\ref{sec:preliminaries}. In Section~\ref{sec:greedy}, we define a new subclass of the \AF algorithms for which we call  the monotone \AF algorithms, and prove an upper bound of the competitive ratio using a new proof technique. In Section ~\ref{sec:HA}, we describe a direct sum property of algorithms for \RSiC and obtain as a corollary an upper bound for $d$-dimensional clairvoyant \RSiC. In Section~\ref{sec:LB},  we prove lower bounds for deterministic and randomized algorithms for \RSiC in both clairvoyant and non-clairvoyant settings by adapting techniques in online graph coloring.  The experimental results are presented in Section~\ref{sec:Experiment}. Finally, we draw the conclusion and list some open problems in Section~\ref{sec:conclusion}.

\section{Previous Work}\label{sec:Previous-Work}

% \lata{Re-organized this section a bit to not repeat as much the definitions already introduced. The references need to be completed and checked. }

Bin packing is one of the best-studied problems in combinatorial optimization, and extensive research has been done on proving tight bounds on the competitive ratio of online algorithms for the problem~\cite{coffman1984approximation, coffman2013bin}. 
A well-studied class of algorithms is \AF, in which the guiding principle is to always use existing bins when possible. Well-known algorithms such as \FF, \BF, and \WF fall into this class. A rather large class of \AF algorithms has been shown to have a tight competitive ratio of 1.7 for the bin packing problem \cite{johnson1974worst}. The best known upper bound for bin packing is achieved by Advanced Harmonic  ~\cite{balogh2017new} and has a competitive ratio of $1.57829$. The current best lower bound is $1.54278$, as shown in~\cite{balogh2021new}.

In the online {\em vector bin packing} problem, item sizes are given as $d$-dimensional vectors and bins have capacity 1 in each dimension. The online algorithm assigns items to bins such that the capacity constraints are respected in every dimension. 
Garey et al.~\cite{garey1976resource} showed that  Generalized  \FF has a competitive ratio at most $d+0.7$. Azar et al.~\cite{azar2013tightVBP}, using connections to online graph coloring, proved a lower bound $\Omega(d^{1-\epsilon})$ for any algorithm. 

Coffman et al.~\cite{coffman1983dynamic} introduced \emph{dynamic bin packing}: items have size as well as duration, and the objective function is still to minimize the number of bins. 
Wong et al. ~\cite{wong20128} established the current best lower bound $2.66$ of the competitive ratio for any algorithm, they also showed that \FF has an upper bound $2.897$. The performance of various algorithms such as \FF, modified \FF, \BF and \WF  have been studied in ~\cite{wong20128,chan2008dynamic,han2010dynamic} etc.

%\yaqiao{I've also edited the following two paragraphs, mainly make it simpler and structured.} \denis{Looks good, removing the comment.}

The $1$-dimensional and $d$-dimensional \RSiC problems were introduced in  Li et al.~\cite{li2015dynamic} and Murhekar et al.~\cite{murhekar2023brief}, respectively. Except ~\cite{murhekar2023brief}, all existing literature on \RSiC studies only dimension 1. 
For  non-clairvoyant $1$-dimensional \RSiC,
Li et al.~\cite{li2015dynamic} proved a lower bound of $\mu+1$ for \AF algorithms.
Later, Kamali and Lopez-Ortiz \cite{RSiC2015MTF} showed that $\mu$ is in fact a lower bound for any deterministic algorithm. The current best upper bound for \FF is $\mu+3$, which was proved in the PhD thesis of Ren ~\cite{ren2018combinatorial}.
Masoori et al.~\cite{masoori2021renting, masoori2021renting_conf} studied the performance ratio of \FF for uniform-duration jobs and for restricted servers and showed better bounds on the competitive ratio of \FF in these settings. 
Kamali and Lopez-Ortiz~\cite{RSiC2015MTF} considered the \NF algorithm and proved that it has competitive ratio of $2\mu + 1$. They also introduced a new \AF algorithm called the {\em Move-to-Front} (\MTF), which places the next job in the {\em most recently used} bin. They proved that the competitive ratio of \MTF is at most $6\mu +7$.
The clairvoyant setting has been studied in \cite{ren2016clairvoyant,Clairvoyant_HA}. In particular, Azar et al.~\cite{Clairvoyant_HA} proposed the hybrid algorithm (\HA) , and proved a tight bound of  $\Theta(\sqrt{\log \mu})$ on its competitive ratio. 

Murhekar et al.~\cite{murhekar2023brief} initiated the study of  non-clairvoyant $d$-dimensional \RSiC. They proved that \MTF has an upper bound $(2\mu +1)d +1$, which in particular improves the previous upper bound of \MTF in \cite{RSiC2015MTF} for $d=1$ to $2 \mu + 2$.
They also generalized various upper and lower bounds of algorithms such as \FF and \NF from dimension 1 to dimension $d$. In particular, they showed a  lower bound of $d(\mu +1)$ for \AF algorithms.

Finally, ~\cite{RSiC2015MTF,ren2016clairvoyant,murhekar2023brief} also presented experiments on random inputs to compare different algorithms and show various non-trivial phenomena. In particular, Kamali and Lopez-Ortiz~\cite{RSiC2015MTF} demonstrated that \MTF in general performs the best among all known non-clairvoyant algorithms at dimension 1, which was recently further confirmed by experiments for dimension $d$ in Murhekar et al.~\cite{murhekar2023brief}.  
% \yaqiao{I added this short paragraph on experimental results. Q: there are some inconsistencies in these experiments in different papers. I think it's very useful to point them out in the experiment discussion section, for reader's (and potentially people who want to do future experiment) benefit. what do you think?} \lata{Yes, I also thought we should mention the experiments here. For the inconsistencies, I kind of agree, I will see how to mention it...}

%\yaqiao{added the follwoing paragraph:} \denis{Looks good, removing the comment.}

To the best of our knowledge, the study of clairvoyant $d$-dimensional \RSiC and randomized algorithms for \RSiC is missing in the literature. This work partially fills this gap.

\section{Notation and Preliminaries}\label{sec:preliminaries}

For $n \in \mathbb{N}$, let $[n]$ denote the set $\{1,2,\cdots, n\}.$ The $L_\infty$ norm of a vector $ v \in \mathbb{R}_{\ge 0} ^d$ is denoted by $\norm{v}_\infty$ and equals $\max_{j \in [d]} v_j$. We shall make frequent use of the following classical inequalities:

\begin{proposition}     \label{prop:norm}
For any set of vectors $v_1, v_2, \cdots, v_n \in \mathbb{R}_{\ge 0} ^d$, we have the following: 
\[
    \norm{\sum_{i=1} ^{n} v_i}_\infty \le  \sum_{i=1} ^{n} \norm{ v_i}_\infty \le  d \cdot  \norm{\sum_{i=1} ^{n} v_i}_\infty 
    \]
\end{proposition}

The input to a $d$-dimensional \RSiC problem is $\cR = \{\sigma_1, \sigma_2, \ldots, \sigma_n \}$ -- a list of jobs  where each job $\sigma_i \in \cR$ is a triple  $(a_i, f_i, \vecs_i)$, denoting the arrival time, finishing time, and size/resource demand of $\sigma_i$. 
We assume that the jobs are presented to the algorithm in the order of arrival, that is, $a_1 \le a_2 \le \cdots \le a_n$.  We refer to $f_i-a_i$ as the \emph{duration} of the job $\sigma_i$. The duration of every job lies between 1 and $\mu$, that is, 
$1 \leq f_i - a_i \leq \mu$ for every job $\sigma_i$. The \emph{utilization} of the job $\sigma_i$, denoted by $\util(\sigma_i)$, is defined as $\util(\sigma_i) = (f_i - a_i) \cdot \norm{\vecs_i}_\infty$.

%\lata{In fact, the size lies between 0 and 1 in each dimension, it's not just a non-negative real.}
Each job has multi-dimensional resource demands, i.e., $\vecs_i \in (0,1)^d$ %\mathbb{R}_{\ge 0} ^d$
where $\vecs_i^j$ denotes the size of the job $\sigma_i$ in the $j^\text{th}$ dimension for $j\in [d]$. We assume that an algorithm for \RSiC  has access to a supply of identical servers of capacity $1$ in each dimension, i.e, the size of each server is $1^{d}$. Thus, for every time $t$ the combined size of jobs in a particular dimension assigned to a particular server and active at time $t$ must not exceed $1$. We denote the sum of sizes of jobs active at time $t$ by $s(\cR, t)$, i.e., $s(\cR, t) = \sum_{i : a_i \le t < f_i} s_i.$ 
For a server $S$ alive at time $t$, we use $S(t)$ to denote the sum of sizes of all jobs that have been scheduled on $S$ and are alive at time $t$. 
As previously mentioned, the cost associated with each server corresponds to the total duration it remains open, and the overall cost of the algorithm is the sum of the costs of all servers. The following example gives a sample input and solution.

\begin{example}  
\label{ex:2d}
%\denis{Perhaps, we should give a $1$-d example instead?}
%\lata{I thought a 2-diml example would be better. The example shows for instance why the second job does not fit even though it "looks" as though it does. }
In this example we have $d = 2$. The input sequence $\sigma$ consists of four jobs, $\sigma_1 = (0, 6, {s_1} = [0.5,0.2]),
\sigma_2 = (1, 4, {s_2} = [0.2, 0.9] ),\sigma_3 = (3, 9, s_3 = [0.2, 0.3] ),   \sigma_4 = (5, 8, {s_4} = [0.6, 0.1])$ . Figure~\ref{fig:cost-example} shows a possible assignment of these jobs to servers. The algorithm opens the first server for $\sigma_1$. However, upon the arrival of $\sigma_2$, a new server is opened since the first server lacks the capacity in the second dimension to accommodate it. When the third job $\sigma_3$ arrives, it is assigned to the first server. However, for $\sigma_4$, a new server must be opened as the first server does not have enough capacity in the first dimension, and the second server is already closed.

The cost of each server is determined by its opening and closing times. The first server's cost is $9$ (opening at $0$ and finishing at $9$), while the second and third servers both have a cost of $3$. The total cost of the algorithm is the sum of the costs of all servers, resulting in a total cost of $15$.
\end{example}

% \begin{figure}[h!]   
% \begin{subfigure}{\linewidth}
%     \centering
%     \includegraphics[width=0.75\textwidth]{Online_Jobs.pdf}
%     \caption{}
%     \label{fig:first}
% \end{subfigure}
% \newline
% \begin{subfigure}{\linewidth}
%     \centering
%     \includegraphics[width=0.75\textwidth]{Cost of Server.pdf}
%     \caption{}
%     \label{fig:second}
% \end{subfigure} 
% \caption{Online assigning the jobs into servers. Note that in this figure, we show the jobs according to the size of the first dimension.}
% \label{fig:cost-example}
% \end{figure}

\begin{figure}[h!]
\begin{center}
\includegraphics[width=0.35\textwidth]{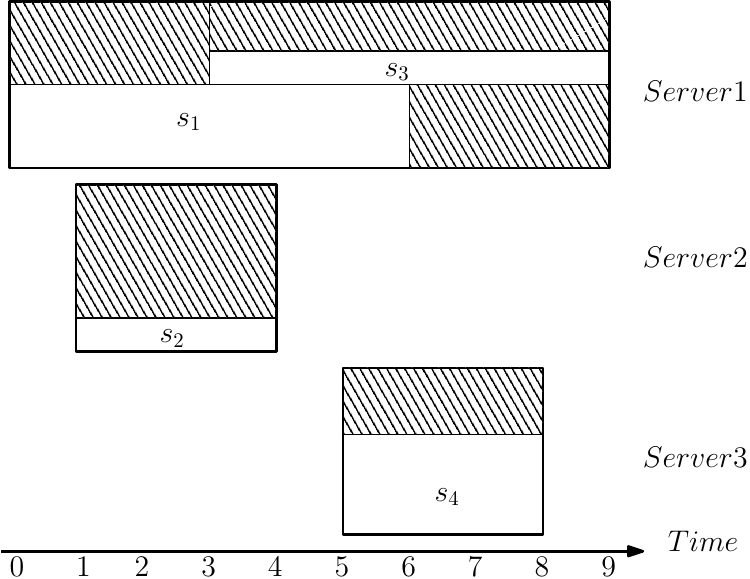}
\caption{Online assignment of jobs into servers described in Example~\ref{ex:2d}. Note that in this figure, we show the jobs according to the size of the first dimension.}
\label{fig:cost-example}
\end{center}
\end{figure}

We shall sometimes also use $r$ to denote an arbitrary job in $\sigma$. In this case, we use the notation
    $r = (a(r),f(r), \vecs(r))$
% for which 
%     $\vecs(r) \in \mathbb{R}^d_\ge 0$,
and we use
    $\vecs(r)_j$
to denote the $j^\text{th}$ coordinate of the size vector $\vecs(r)$.
The two key parameters of an instance $\cR$, usually called the \emph{span} and the \emph{utilization}~\cite{li2015dynamic}, are defined respectively below,
\[
    \myspan(\cR) = \left| \cup_{i \in [n]} [a_i, f_i)\right |, \quad 
    \util(\cR) = \sum_{i =1} ^ {n} \util(\sigma_i).
\]
Without loss of generality, we assume that $\cup_{i \in [n]} [a_i, f_i) = [0, T)$, that is, $\myspan$ arises from a single uninterrupted interval, and the first job arrives at time $0$.

For an algorithm $\ALG$ and $t \in [0,T)$ we use $\ALG(\cR, t)$ to denote the number of servers opened by $\ALG$ that are active at time $t$. We use $\ALG(\cR)$ to denote the total cost of $\ALG$ on input $\cR$, i.e., the sum of durations of servers opened by $\ALG$. Similar notation is used for $\OPT$. Observe that

\begin{proposition}\label{Prop:opt-alg-lb}
%{\ } \\
\[\OPT(\cR) = \int_{0} ^ {T} \OPT(\cR, t) dt, \text{ and} \]
\[\ALG(\cR) = \int_{0} ^ {T} \ALG(\cR, t) dt.\]
% $\OPT(\cR) = \int_{0} ^ {T} \OPT(\cR, t) dt, \text{ and }
% \ALG(\cR) = \int_{0} ^ {T} \ALG(\cR, t) dt.$
\end{proposition}

We note the following general lower bound on $\OPT(\cR)$.

\begin{lemma}\label{lemma:opt_LB}
 % $\int_{0}^T \norm {s(\cR, t)}_\infty dt \le d \cdot OPT(\cR) $ 
  $\int_0^T \lceil\norm {s(\cR, t)}_\infty \rceil dt \le \OPT(\cR) $ 
\end{lemma}

\begin{proof}
    As the capacity of each server is $1$ for every dimension $j \in [d]$; any algorithm needs at least $\lceil\norm {s(\cR, t)}_\infty \rceil$ servers to pack the total load at any time $t$.
    %in the $j^\text{th}$ dimension.
    Therefore, $\OPT(\cR,t) \ge 
    %\lceil s(\cR, t) ^{j} \rceil =  
    \lceil\norm {s(\cR, t)}_\infty \rceil$. Using Proposition \ref{Prop:opt-alg-lb}, we can conclude: 
       \[
       \int_{0}^T \lceil\norm {s(\cR, t)}_\infty \rceil dt \le  \int_{0} ^ {T} \OPT(\cR, t) dt = \OPT(\cR). \qedhere
       \]
\end{proof}

%The following proposition gives lower bounds on the optimal cost.

\begin{corollary}[\cite{murhekar2023brief}] \label{prop:OPT-bounds}
      $\OPT(\cR) \ge \max\{\myspan(\cR), \util(\cR)/d\}$.
\end{corollary}

An online algorithm $\ALG$ is said to be \emph{asymptotically} $\rho$-competitive if there exists a constant $c>0$ such that for all input sequences $\sigma$:
\begin{equation}\label{eq:competitiveness}
    \ALG(\sigma) \le \rho \cdot \OPT(\sigma) + c.
\end{equation} 
The infimum over all such $\rho$ is denoted by $\rho(\ALG)$ and is called the \emph{competitive ratio} of $\ALG$. If $c = 0$ then the algorithm is called \emph{strictly} $\rho$-competitive.

Lastly, we use  notation $\mathbbm{1}(C)$ for the indicator function that evaluates to $1$ if the condition $C$ is satisfied, and it evaluates to $0$ otherwise.

\section{Monotone \AF Algorithms}  \label{sec:greedy}

An algorithm for \RSiC is said to be an \AF algorithm if it opens a new server only in case that a new incoming job cannot be accommodated on any of the currently active servers. Most \AF algorithms use an ordering of active servers, and assign the next job to the first server in the ordering with enough available space to accommodate the job. In this case, we say that the algorithm \emph{employs an ordering}. For example, \FF orders servers based on their opening times, \BF  orders servers based on their remaining capacity, and \MTF moves the server to which a job is assigned to the first position in the ordering. Observe that the ordering of servers could be fixed as in \FF and \LF, or it could change when a job arrives, as in \BF, \WF, and \MTF, as well as when a job leaves, as in \BF and \WF. 

Consider an algorithm $\ALG$ that employs an ordering. We say that a server $S$ is higher in the ordering than $S'$ at time $t$ if $S$ appears closer to the beginning of the ordering than $S'$. Consider $t < t'$ and define $A(t, t')$ to be the set of servers that are alive at $t$ and $t'$.%Let $t < t'$ be two times and consider a server $S$ that is alive at $t$ and $t'$. We say that $S$ moved up (respectively, down) in the ordering between $t$ and $t'$ if the position of $S$ at $t'$ is closer to the beginning (respectively, the end) of the list than the position of $S$ at $t$. L
\begin{definition} \label{def:monotone}
    An \AF algorithm $\ALG$ is called \emph{monotone} if 
    \begin{itemize}
        \item it employs an ordering, and
        \item for every $t < t'$ and every server $S \in A(t, t')$: if $S$ did not receive any new jobs during the interval $(t, t')$ then every server in $A(t, t')$ that is higher than $S$ in the ordering at time $t$ is still higher than $S$ in the ordering at time $t'$.
    \end{itemize}
\end{definition}
Note that for a monotone \AF algorithm a server $S$ can move up in the ordering between $t$ and $t'$ only if either some server that was higher than $S$ at time $t$ was released before $t'$, or $S$ received a job during the interval $(t, t')$.  
%\lata{I added more explanation of the monotone property below, we can remove/rewrite, saying it is obvious that \FF,\LF, and \MTF are monotone and \MTF is not,  if we need space. } \denis{Looks good, removing the comment.}
It is clear that \AF algorithms employing a static ordering such as \FF and \LF have the monotone property. Observe that in \MTF, a server that receives a job moves to the first position in the ordering, and the relative position of other servers stays the same. Since the only way for a server to move ahead of other servers in the ordering is for it to receive a job, \MTF obeys the monotone property. However, in \BF, a server may move down in the ordering when a job departs, causing the server to have more available space. Thus \BF does not obey the monotone property.

We propose a new monotone \AF algorithm, that surprisingly has not been studied earlier: 

% \lata{Mahtab noted that \LF was already proposed and implemented by Muhrekar et al. 
%\noindent {\bf \LF:} Orders servers in decreasing order of their opening times.
\vspace*{0.1in}

\noindent {\bf \Greedy:} Order servers in decreasing order of their {\em finishing times}, that is, the maximum of the finishing times of jobs currently in the server. Assign the newly arrived job to the first server in the order that has sufficient capacity. If no such server exists, open a new server and assign the job to it.
\vspace*{0.1in}

\Greedy is a natural and easy-to-implement algorithm that uses the greedy heuristic of assigning the incoming job to the server that will incur the {\em least additional cost} to the algorithm. Observe that a server $S$ moves ahead of another server $S'$ in the ordering if and only if $S$ receives a job that causes its finishing time to be higher than that of $S'$.
Therefore \Greedy obeys the monotone property specified in Definition~\ref{def:monotone}. We note that it is a clairvoyant algorithm.

We need the following lemma before we prove the main result of this section. We denote the sum of $L_\infty$ norms of sizes of jobs with arrival time in the interval $(t, t')$ for $t < t'$ by $s_\infty(\cR, t, t')$, i.e., $s_\infty(\cR, t, t') = \sum_{i : t < a_i < t'} \norm{s_i}_\infty.$ 

\begin{lemma}\label{lemma:alpha-norm}
    $\int_{0}^T s_\infty(\cR, t-\alpha, t) dt = \alpha \sum_{i = 1}^n \norm{s_i}_\infty.$  
\end{lemma}

\begin{proof}
\begin{align*}
    \int_{0}^T s_\infty(\cR, t-\alpha, t) dt &= \int_{0}^T \sum_{i = 1}^n \mathbbm{1}(t-\alpha < a_i < t) \norm{s_i}_\infty dt\\
    &= \sum_{i = 1}^n  \int_{0}^T \mathbbm{1}(t-\alpha < a_i < t) \norm{s_i}_\infty dt\\
    &= \sum_{i = 1}^n \norm{s_i}_\infty \int_{0}^T \mathbbm{1}(t-\alpha < a_i < t)  dt\\
    &= \sum_{i = 1}^n \alpha \norm{s_i}_\infty. \qedhere
\end{align*}
\end{proof}

% \begin{lemma}\label{lemma:norm_to_opt}
%     $\sum_{i = 1}^n \norm{s_i}_\infty \le d \OPT(\cR).$  
% \end{lemma}
% \begin{proof}
%     \begin{align*}
%         \sum_{i = 1}^n \norm{s_i}_\infty &\le \sum_{i = 1}^n \norm{s_i}_\infty \int_{a_i}^{f_i}1 dt\\
%         &= \sum_{i = 1}^n \norm{s_i}_\infty \int_0^T I(a_i < t < f_i) dt \\
%         &= \int_0^T \sum_{i = 1}^n \norm{s_i}_\infty I(a_i < t < f_i) dt \\
%         &= \int_0^T \sum_{i : a_i < t < f_i} \norm{s_i}_\infty dt \\
%         &\le \int_0^T d \norm{s(\cR, t)}_\infty dt \\
%         &\le d \OPT(\cR),
%     \end{align*}
%     where the second to last inequality follows from Proposition~\ref{prop:norm}, and the last inequality follows from Lemma~\ref{lemma:opt_LB}.
% \end{proof}

Now, we are ready to prove the main result of this section.

\begin{theorem}\label{Thm: monotone-proof}
    Let $\ALG$ be a monotone \AF algorithm. Then we have $\rho(\ALG) \le 3\mu d +1$.
\end{theorem}

\begin{proof}
We claim that for an arbitrary $t$ it holds that
\begin{equation}\label{eq:monotone_main_ineq}
    \ALG(\cR, t) \le s_\infty(\cR, t-2\mu, t) + s_\infty(\cR, t-\mu, t) + 1.
\end{equation}
Observe that each server in $A(t-\mu, t)$ must have received a job during the time interval $(t-\mu, t)$, otherwise a server alive at time $t-\mu$ would have been released by time $t$, since the duration of each job is at most $\mu$. Suppose there are $q$ servers in $A(t-\mu,t)$ named $A_1, A_2, \ldots, A_q$, ordered according to the ordering of $\ALG$ at time $t-\mu$. Let $t_i$ be the earliest time in $(t-\mu, t)$ when a job with the size vector $s_i$ arrived in server $A_i$. Let $B(t-\mu, t)$ denote the set of new servers that were opened during time $(t-\mu, t)$ that are still alive at time $t$. Suppose there are $p$ such servers called $B_1, B_2,\ldots, B_p$ ordered by their opening times $t'_1, t'_2, \ldots, t'_p$.  Let $s'_i$ be the size vector of the first job placed into $B_i$. See Figure~\ref{fig:monotone} for an illustration.
Note that we have $\ALG(\cR, t) = p + q$.

\begin{figure}[t!]
\begin{center}
\includegraphics[width=0.5\textwidth]{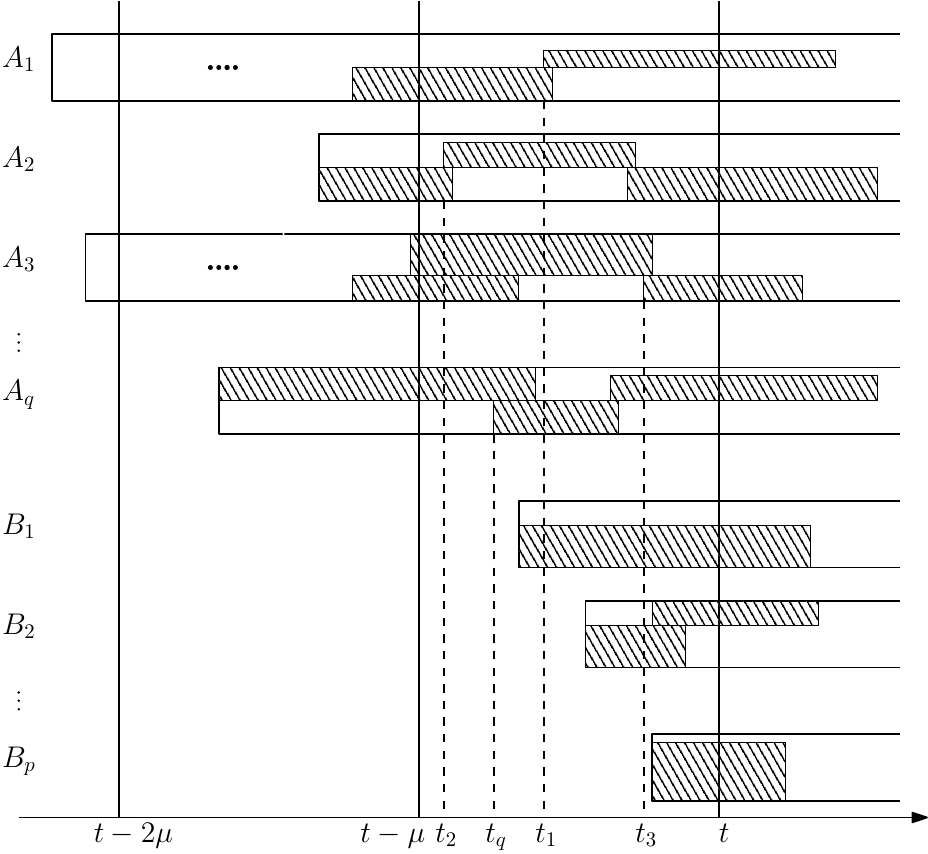}
\caption{Configuration of servers in interval $[t-2\mu$, $t$]. Note that the $A_i$ servers are ordered according to the ordering of $\ALG$ at time $t - \mu$ while the $B_i$ servers are ordered by opening time.  }
\label{fig:monotone}
\end{center}
\end{figure}

Consider some $i \in \{2, \ldots, q\}$. Observe that $A_{i-1}$ preceded $A_i$ in the ordering of $\ALG$ at time $t-\mu$, no job arrived in $A_i$ during time interval $(t-\mu, t_i)$, and $\ALG$ is monotone. Thus, $A_{i-1}$ precedes $A_i$ in the ordering of $\ALG$ immediately prior to arrival of job $s_i$. Thus, $\ALG$ must have tried placing $s_i$ into server $A_{i-1}$ at time $t_i$, but could not fit it in (since $s_i$ was ultimately placed into $A_i$). This happened because in some coordinate the total size of jobs in server $A_{i-1}$ plus the size of the job $s_i$ in that coordinate exceeded the capacity. Thus, we can conclude that
\begin{equation} \label{eq:a-servers}
    \norm{s_i + A_{i-1}(t_i)}_\infty > 1.
\end{equation}
For the $B_i$ servers,  since $\ALG$ is an \AF algorithm (it opens a new server only if it has to), we have:
\begin{equation} \label{eq:b-servers}
\norm{s'_1 + A_q(t'_1)}_\infty > 1 \mbox{ and  }\norm{s'_i + B_{i-1}(t'_i)}_\infty > 1 \mbox{ for }2 \leq i \leq p
\end{equation}

Also, note that 
\begin{equation}\label{eq:monotone_helper1}
\sum_{i = 2}^q \norm{A_{i-1}(t_i)}_\infty+ \norm{A_q(t'_1)}_\infty + \sum_{i=2}^p\norm{B_{i-1}(t'_i)}_\infty \le s_\infty(\cR, t-2\mu, t),
\end{equation}
since all jobs that are alive in server $A_{i-1}$ at time $t_i$, as well as in $A_q$ at time $t'_1$, must have arrived between $t-2\mu$ and $t$, and the $A_i$ and the $B_j$ servers partition the set of relevant jobs. In addition, we have 
\begin{equation}\label{eq:monotone_helper2}
\sum_{i = 2}^q \norm{s_i}_\infty  + \norm{s'_1}_\infty +\sum_{i=2}^p \norm{s'_i }_\infty \le s_\infty(\cR, t-\mu, t),
\end{equation}
since the $s_i$ and the $s'_i$ jobs have arrival times between $t-\mu$ and $t$, and the jobs are distinct.

Combining the observations in (\ref{eq:a-servers}) and (\ref{eq:b-servers}), we obtain 
\begin{align*}
    q+p-1 &< \sum_{i = 2}^q \norm{s_i + A_{i-1}(t_i)}_\infty + \norm{s'_1 + A_q(t'_1)}_\infty\\
    &+ \sum_{i=2}^p \norm{s'_i + B_{i-1}(t'_i)}_\infty\\
    &\le s_\infty(\cR, t-2\mu, t) + s_\infty(\cR, t-\mu, t),
    % &\le d \norm{\sum_{i = 2}^q s_i + A_{i-1}(t_i) + s'_1 + A_q(t'_1) +\sum_{i=2}^p s'_i + B_{i-1}(t'_i) }_\infty\\
    % &\le d \norm{\sum_{i = 2}^q A_{i-1}(t_i)+ A_q(t'_1) + B_{i-1}(t'_i)}_\infty \\
    % &+ d \norm{\sum_{i = 2}^q s_i  + s'_1 +\sum_{i=2}^p s'_i }\\
    % &\le d \norm{s(\cR, t-2\mu, t)}_\infty+d \norm{s(\cR, t-\mu, t)}_\infty,
\end{align*}
where the second inequality follows from Proposition~\ref{prop:norm} and application of ~\eqref{eq:monotone_helper1} and \eqref{eq:monotone_helper2}. This establishes Inequality~\eqref{eq:monotone_main_ineq}.

To finish the proof of the theorem, we integrate this  inequality over possible values of $t$, i.e.:
 \begin{align*}
     \ALG(\cR) &= \int_0^T \ALG(\cR, t) dt \\
     &\le \int_0^T ( s_\infty(\cR, t-2\mu, t) + s_\infty(\cR, t-\mu, t) + 1 )dt \\
     &= 2\mu \sum_{i = 1}^n \norm{s_i}_\infty + \mu \sum_{i = 1}^n \norm{s_i}_\infty  + \myspan(\cR)\\
     &\le 3 \mu \util(\cR) + \myspan(\cR)\\
     &\le 3\mu d \OPT(\cR) + \OPT(\cR),
 \end{align*}
 where the second equality follows from two applications of Lemma~\ref{lemma:alpha-norm}, and the last inequality is an application of Proposition~\ref{prop:OPT-bounds}.
\end{proof}

We note that there is not much room for improvement of the bound in the above theorem, because Murhekar et al.~\cite{murhekar2023brief} proved a lower bound of $(\mu +1)d$ for any \AF algorithm.

\begin{corollary}
\label{cor:greedy}
    \[(\mu+1)d \le \rho(\Greedy), \rho(\LF) \le 3d\mu+1\]
\end{corollary}

As mentioned earlier,  the proof of a $6\mu+8$ upper bound on the competitive ratio of \MTF given in \cite{RSiC2015MTF} also works for \Greedy in the 1-dimensional case. However, Corollary~\ref{cor:greedy} improves this bound. We conjecture that the correct bound on the competitive ratio of \Greedy is $\mu d + O(d)$.

\section{A Direct-Sum Property of \RSiC} \label{sec:HA}

%In this section we shown a $O(d\sqrt{\log \mu})$ upper bound for $d$-dimensional \RSiC, by giving a modified version of the $\HA$ algorithm that was initially proposed in \cite{Clairvoyant_HA} for the case $d=1$. For completeness, we briefly recall how the algorithm $\HA$ works. The algorithm classifies a job $r$ to be of type $T=(i,c)$ if
%\begin{equation}    \label{eq:HA_type}
%    \ell(r) \in [2^{i-1},2^i), \quad
%    a(r) \in [(c-1)2^i, c 2^i)
%\end{equation}
%for 
%    $1 \le i \le \log \mu$
%and 
%    $c \in \bbN$.
%The algorithm also maintains two categories of bins, called $GN$ and $CD$ bins, respectively. Each $CD$-type bin may only hold jobs of the same type, while $GN$ bins can hold any job. Upon the arrival of a job $r$ of type $T$, \HA checks if there is an  open $CD$-type bin that holds type $T$ jobs. If such a bin exists,
%\HA will pack $r$ into one of these bins in a \FF manner (opening a new CD-type bin if needed).
%Otherwise, it checks if the sum of sizes of all active jobs of type $T$ (including the size of the
%current job $r$) is strictly greater than 
%    $\frac{1}{2\sqrt{i}}$.
%If so, it opens a new $CD$-type bin and packs $r$ into it; otherwise, it packs $r$ into a $GN$-type bin in a \FF manner (opening a new GN-type bin if needed).

Given an arbitrary algorithm $\cA$ for $1$-dimensional  \RSiC, we define an algorithm, call it $\cA^{\oplus d}$, that works for  $d$-dimensional \RSiC, as follows. Let $\sigma$ be an input instance for the $d$-dimensional problem. We partition $\sigma$ as follows:
\[    \sigma = \sigma^{(1)} \cup \cdots \cup \sigma^{(d)},\]
where
    $\sigma^{(j)}$
is the subset of jobs $r$ for which 
    $\norm{\vecs(r)}_\infty$ 
is achieved at the $j$-th dimension. When $\norm{\vecs(r)}_\infty$ is achieved in more than one dimension, we break the tie arbitrarily. It is easy to see that this partitioning can be done online. Each $\sigma^{(j)}$ will be assigned to a different set of servers. 

% {\noindent\bf Observation 1: } Every $\sigma^{(j)}$ can be constructed online. 

The algorithm $\cA^{\oplus d}$ is defined as follows: on the arrival of a job $r$, decide online a unique dimension $j$ in which its size is maximum, and assign $r \in \sigma^{(j)}$. Next apply the algorithm $\cA$ (for 1-dimensional \RSiC) to process $\sigma^{(j)}$, pretending that the instance is $1$-dimensional by only looking at the size at the $j$-th coordinate, and ignoring the sizes of other dimensions, and assigning to servers that only contain jobs in $\sigma^{(j)}$.

%{\noindent\bf Observation: } Let $B$ be a server created by $\cA^{\oplus d}$ for  $\sigma^{(j)}$.%, let 
    %$\vecs(B,t) = \sum_r \vecs(r)$
%where $r$ ranges over all jobs in $B$ that are active at time $t$.
% Then, for any time $t$,
%     $\norm{B(t)}_\infty$
% is achieved at dimension $j$.

%Let $p(\cdot)$ denote the competitive ratio for an arbitrary (either for dimension $1$ or for dimension $d$.) algorithm.

\begin{theorem} \label{thm:direct_sum}
    Let $\cA$ be an arbitrary deterministic  algorithm for $1$-dimensional \RSiC. Then, $\cA^{\oplus d}$  works correctly for any $d$-dimensional \RSiC, and 
        $\rho(\cA^{\oplus d}) = d\cdot \rho(\cA)$. 
    Moreover, the guarantee on the competitive ratio holds for both strict and asymptotic competitive ratios.
\end{theorem}

\begin{proof}
    Firstly, we show that the algorithm $\cA^{\oplus d}$ does not violate the size constraint, i.e., the total size of all jobs in every server does not exceed $1^d$. 
    To see this, consider an arbitrary  job $r \in \sigma$ and suppose it is put into a bin $B$ by $\cA^{\oplus d}$. By the definition of $\cA^{\oplus d}$, we know before $r$ is put into $B$, either server  $B$ is empty (i.e., has not been created yet) in which case after $r$ is assigned to server $B$ the size constraint is trivially respected, or server $B$ is nonempty. In the latter case,  it only contains jobs in $\sigma^{(j)}$. In this case, we have
    \[
        \norm{\vecs(r) + B(t)}_\infty 
        = \vecs(r)_j + B(t)_j \le 1,
    \]
    where the equality follows from the fact that every job in server $B$
is in $\sigma^{(j)}$
    and also $r \in \sigma^{(j)}$. The inequality follows by the fact that we apply algorithm $\cA$ on $r \in \sigma^{(j)}$.

    Next, we show 
        $\rho(\cA^{\oplus d}) \le d\cdot \rho(\cA)$.
    By an abuse of notation, let $\cA(\sigma^{(j)})$ denote the cost of $\cA^{\oplus d}(\sigma)$ on the subset of inputs $\sigma^{(j)}$. Since
        $\sigma^{(j)} \subseteq \sigma$,
    one has
        $\OPT(\sigma^{(j)}) \le \OPT(\sigma)$
    for every $j$. 
    Let $\OPT'(\sigma^{(j)})$ denote the cost of the optimal solution that processes $\sigma^{(j)}$ by only focusing on the size of the $j$-th dimension. Then, for every $\rho > \rho(\cA)$ there exists a $c > 0$ such that
        $\cA(\sigma^{(j)}) \le \rho \cdot \OPT'(\sigma^{(j)}) + c$
    for every $j$.
    Observe that we have
        $\OPT'(\sigma^{(j)}) = \OPT(\sigma^{(j)})$.
    This is because every job in $\sigma^{(j)}$ satisfies that the size at the $j$-th dimension is the largest. With these, and by the definition of $\cA^{\oplus d}$, we have
    \begin{align*}
        \cA^{\oplus d}(\sigma)
        &= \sum_{j \in [d]} \cA(\sigma^{(j)}) \\
        &\le \sum_{j \in [d]}\left( \rho \cdot \OPT'(\sigma^{(j)})  + c \right) \\
        &= \sum_{j \in [d]}\left(\rho \cdot \OPT(\sigma^{(j)}) + c \right) \\
        &\le \sum_{j \in [d]}\left(\rho \cdot \OPT(\sigma) + c \right)
        = d \cdot \rho \cdot \OPT(\sigma) + cd. 
    \end{align*}
    Since $cd$ is a constant independent of input, and this inequality holds for all $\rho > \rho(\cA)$, it follows that $\rho(\cA^{\oplus d}) \le d \rho(\cA)$. Moreover, if $c = 0$ then $cd = 0$, so the competitive ratio guarantee preserves strictness.

    Lastly, we show 
        $\rho(\cA^{\oplus d}) \ge d\cdot \rho(\cA)$.
    Let $H$ be an arbitrary $1$-dimensional instance, from which we construct a $d$-dimensional instance $\sigma$ as follows. For every job 
        $h=(a(h), f(h), s(h)) \in H$,
    create $d$ jobs in $\sigma$ that have the same arrival and finishing time as $h$, and the size vectors are the $d$ column vectors of the matrix
        $s(h) \cdot I_d$
    where $I_d$ is the $d\times d$ identity matrix. Clearly, every $\sigma^{(j)}$ is simply a copy of $H$ in dimension $j$, while having $0$'s in all other dimensions. Hence, 
        $\ALG^{\oplus d}(\sigma) = d \cdot \ALG(H)$.
    Furthermore, Observe that
        $\OPT(H) = \OPT(\sigma)$,  
    where here by an abuse of notation we use $\OPT(H)$ to denote the cost of the optimal algorithm for the $1$-dimensional instance $H$, and $\OPT(\sigma)$ to denote the cost of the optimal algorithm for the $d$-dimensional instance $\sigma$. Hence,
    \[
        \rho(\ALG^{\oplus d}) 
        \ge \frac{\ALG^{\oplus d}(\sigma)}{\OPT(\sigma)}  
        = \frac{d \cdot \ALG(H)}{\OPT(\sigma)}  
        = d \cdot \frac{\ALG(H)}{\OPT(H)}.
    \]
    Because $H$ is arbitrary, the desired lower bound follows.
\end{proof}

In \cite{Clairvoyant_HA}, the Hybrid Algorithm $\HA$ is defined for $1$-dimensional clairvoyant \RSiC, and is shown to have a competitive ratio $\Theta(\sqrt{\log \mu})$. 

\begin{corollary}   \label{cor:HA}
    The algorithm $\HA^{\oplus d}$ for $d$-dimensional clairvoyant \RSiC has a competitive ratio 
        $\Theta(d\sqrt{\log \mu})$.
\end{corollary}

\section{Lower bounds via online graph coloring}    \label{sec:LB}

 In this section, we show lower bounds on the competitive ratio of any algorithm for $d$-dimensional \RSiC. We consider deterministic algorithms in Section~\ref{sec:LB_deterministic} for both the clairvoyant and non-clairvoyant versions of the problem. In Section~\ref{sec:LB_randomized},
 we give lower bounds for randomized algorithms for both versions of the problem.
 
\subsection{Deterministic algorithms}   \label{sec:LB_deterministic}

For $1$-dimensional non-clairvoyant \RSiC, it is shown \cite[Theorem 1]{RSiC2015MTF}  that $\mu$ is a lower bound for \emph{any} deterministic algorithm. By adapting the online graph coloring lower bound construction from Halld{\'o}rsson-Szegedy \cite{HScoloring}, we show a $\widetilde{\Omega}(d\mu)$
lower bound for  $d$-dimensional non-clairvoyant \RSiC. For readers familiar with the online graph construction in \cite{HScoloring}, a size vector in the proof of Theorem \ref{thm:LB_dimd} below corresponds to a vertex in the online graph, and the size vectors of the jobs are chosen in the way such that jobs that fit into one server correspond to an independent set of vertices of the online graph. The idea of applying online graph coloring to online vector bin packing has been discussed in \cite{azar2013tightVBP,LI202370VBP}.

\begin{theorem} \label{thm:LB_dimd}
    There exists a constant $d_0$ (which can take $d_0=280$) such that for every dimension
        $d \ge d_0$ 
    and 
        $\mu \le \log d - 2$,
    every deterministic algorithm for non-clairvoyant \RSiC has a competitive ratio $\ge \frac{d}{ \log^3 d} \cdot \mu$.
    For clairvoyant \RSiC, the lower bound 
        $\Omega( \max\{\sqrt{\log \mu}, \frac{d}{ \log^2 d}\})$
    holds.
\end{theorem}

\begin{proof}
    Let $\ALG$ be an arbitrary deterministic algorithm for $d$-dimensional non-clairvoyant \RSiC. 
Choose $k$ to be the largest integer for which 
        $d \ge d' = {2k \choose k} \cdot k \ge 2^{2k}$.
     Let $\ADV$ denote the adversary, which will adaptively construct an instance 
        $\sigma = (\sigma_1, \ldots, \sigma_{d'})$
    consisting of $d'$ jobs, such that $\ALG$ uses at least
        $d'/k$
    servers on $\sigma$, while $\OPT$ uses at most $2k$ servers. In fact, $\ADV$ will construct a solution using no more than $2k$ servers.     Note that all jobs arrive at the same time $0$. At the end of $d'$ steps, once all jobs have been processed by $\ALG$, the adversary specifies the finish times of all jobs.

    % Let
    %     $\sigma_i = (a_i, f_i, \vecs_i)$,
    % where
    %     $\vecs_i \in [0,1]^d$.
    % Since it is non-clairvoyant, the adversary will choose some $f_i$'s to be $1$, and some to be $\mu$, to be determined later. The main task for the adversary is to construct the size vectors
    %     $\vecs_1, \ldots, \vecs_d$.
    % Because all $a_i=0$, for simplicity we often use $s_i$ to mea job $\sigma_i$. For example, when we say $\ALG$ assigns $s_i$ to a server we meant it assigns job $\sigma_i$ in that server.

    % Let $k$ be a parameter to be determined later. The goal of the (adaptive) construction is the following.

    % {\bf\noindent Goal:} $\ALG$ uses at least
    %     $d/k$
    % servers on $\sigma$, while $\OPT$ uses at most $2k$ servers. 
    
    % In fact, the adversary will construct a solution using at most $2k$ servers. 
    To avoid confusion, we call the servers used by $\ALG$ {\em bins}, while the servers used by $\ADV$ will be called {\em servers}. We name the $2k$ servers that will be used by $\ADV$ as $1,2,\ldots,2k$. For an arbitrary job $r \in \sigma$, 
 %   let     $\ALG(r)$ denote the bin to which $\ALG$ assigns $r$, and 
   let  $\ADV(r) \in [2k]$ denote the server to which $\ADV$ assigns $r$. 
    Let $B$ be an arbitrary bin used by $\ALG$, and let $\hat{B}_{\le i}$ denote the set of jobs in $B$ after the first $i$ jobs have been processed, and let $B_{\le i}$ denote the sum of the sizes of these jobs. %\denis{this redefines the notation $B(t)$ defined in the preliminaries, where $t$ is a time.}
    %\denis{Perhaps, we can use $i$ instead of $t$, and $\hat{B}_{\le i}$ instead of $\hat{B}(t)$, and $B_{\le i}$ instead of $B(t)$.}
     Define
    \[
        \ADV(B, i) := \{\ADV(r): r \in \hat{B}_{\le i}\} ,
    \]
    i.e., it is the set of servers that the adversary used to assign jobs present in bin $B$ after the first $i$ jobs have been processed. 
    %Let
   %     $|B|$
 %   denote the number of jobs in bin $B$,
    %
      %  $|\ADV(B,t)|$
   % denote the number of servers in $\ADV(B)$.     
    
%   So, after $\ALG$ assigned job $\sigma_t$ to some bin, $\ADV$ will assign $\sigma_t$ to some server.

    In step $i$, the adversary specifies the job $\sigma_i$, whose size is defined as follows. Let $X$ denote the set of servers used by $\ALG$ that have exactly $k$ jobs placed in them after the first $i-1$ jobs have been processed. We define 
    \begin{equation}    \label{eq:def_Ft}
        \cF_i = \{\ADV(B, i-1): B \in X \}.
    \end{equation}
    Note that $\cF_i$ is a set of subsets of $[2k]$; each element of $\cF_i$ is the set of servers in which the adversary placed the jobs of $k$-sized bins of $\ALG$.  
    The adversary chooses an arbitrary subset $A_i$ of size $k$
    \begin{equation}    \label{eq:choose_At}
        A_i \in {[2k] \choose k} - \cF_i,
    \end{equation}
  %  The meaning of $A_t$ is that adversary will assign $s_t$ to one of the $k$ servers in $A_t$ (so, the letter $A$ is chosen to mean available). 

Since each distinct subset of $[2k]$ in $\cF_i$ corresponds to $k$ jobs in the input sequence $\sigma$, and recalling that the length of $\sigma$ is $d'={2k \choose k}k$, the set $A_i$ always exists.

The adversary now defines $\vecs_i \in [0,1]^{d'}$ as follows: 
    \begin{equation}    \label{eq:def_s}
        \vecs_i^j =
        \begin{cases}
            0, &\quad \text{if $j<i$ and $\ADV(\sigma_j) \in A_i$}, \\
            1/d, &\quad \text{if $j<i$ and $\ADV(\sigma_j) \not\in A_i$}, \\
            1, &\quad \text{if $j=i$}, \\
            0, &\quad \text{if $j>i$}. 
        \end{cases}    
    \end{equation}

The adversary assigns $\sigma_i$ to an arbitrary server $Q \in A_i$. We show below that this is a valid assignment, that is, it does not violate capacity constraints.  

\begin{claim} \label{valid-assignment}
%The assignment of $\sigma_t$ to $Q \in A_t$ is a valid assignment. 
Let 
        $Q \in A_i$
    denote an arbitrary server in $A_i$. Then
    \begin{equation}    \label{eq:adv_sln_correct}
        \norm{\vecs_i + Q(i-1)}_\infty \le 1
    \end{equation}
\end{claim}
  \begin{proof}  
  Indeed, by \eqref{eq:def_s},
         $\sigma_j \in \hat{Q}(i-1)$
     implies
         $\vecs_i^j = 0$.
     Also, because $j<i$, again by \eqref{eq:def_s} one has $\vecs_j^i = 0$.
 %   Hence, \eqref{eq:adv_sln_correct} holds. 
 The claim follows.
    \end{proof}
  The adversary now presents job $\sigma_i$ to $\ALG$ and suppose $\ALG$ assigns $\sigma_i$ to bin $B$. Next, we show that after the assignment, bin $B$ (and indeed any bin of $\ALG$) contains at most $k$ jobs.

\begin{claim} \label{size-at-most-k}
$|\hat{A}(i)|\leq k$ for every bin $A$ in use by $\ALG$ after $i$ jobs have been processed. 
\end{claim}
\begin{proof}
Assume this is true inductively after $i-1$ jobs have been processed. Since $\sigma_i$ was added to bin $B$, we only need to argue that bin $B$ had at most $k-1$ jobs after $i-1$ jobs were processed. Suppose instead that $|\hat{B}(i-1)| = k$. Then $\ADV(B, i-1) \in \cF_i$, and so $A_i \neq \ADV(B, i-1)$. 
%
% it can be seen from the choice of $C_t$ that for any two distinct jobs $r$ and $r'$ in $B_t$, we have $\ADV(r) \new \ADV(r')$. This further implies $|ADV(B_t)| = |B_t|$.

% It suffices to show that if a bin $B$ used by $\ALG$ has reaches $|B|=k$, then no more job can be put into $B$ by $\ALG$. Indeed, by Property 1, we know that
%         $|\ADV(B)| = |B| = k$.
%     Hence,
%         $A_t \neq \ADV(B)$.
    Since both $A_i$ and $\ADV(B, i-1)$ are of size $k$, there must exist at least one job
        $\sigma_j \in \hat{B}(i-1)$
    for which
        $\ADV(\sigma_j) \not\in A_i$.
    This implies
        $\vecs_i^j + \vecs_j^j = 1/d + 1 > 1$.
    Hence,
        $\vecs_i$
    does not fit into bin $B$, which contradicts the choice of $B$ by $\ALG$.
    \end{proof}

Claim~\ref{valid-assignment} shows that the optimal assignment uses at most $2k$ servers while Claim~\ref{size-at-most-k} shows that $\ALG$ needs at least $d'/k$ bins to assign the jobs in $\sigma$.

    We are now ready to determine the choice of $f_i$, the finishing time of job $\sigma_i$. Observe that there exists at least one server of the adversary,  call it $Q$, such that jobs in $\hat{Q}_{\le d'}$ appear in at least
        $(d'/k)/(2k) = d'/2k^2$
    distinct bins of $\ALG$. 
    The adversary sets the duration of jobs $\hat{Q}_{\le d'}$ to be $\mu$, and the duration of all other jobs to be $1$.  
    
    Hence, the cost of $\OPT$ is at most
        $2k-1 + \mu$,
    while the cost of $\ALG$ is at least
        $d'\mu /2k^2$. Using the fact that for $k \geq 4$, we have 
        $d' = {2k \choose k} \cdot k \ge 2^{2k}$,  the competitive ratio of $\ALG$  is at least    

%     It remains to choose $k$.
%     By the adversary's construction \eqref{eq:choose_At}, the adversary can present at least 
%         ${2k \choose k} \cdot k$
%     size vectors, where ${2k \choose k}$ is the number of all subsets of size $k$ in $[2k]$, and every $k$-subset (of servers)  becomes unavailable only after at least $k$ vectors created. 
% \lata{I think it would be better to write this part just for arbitrary d, but haven't done it yet.}    
    % there exists at least one adversary's server, call it $q^*$, such that jobs in $q^*$ appear in at least
    %     $(d/k)/(2k) = d/2k^2$
    % distinct bins of $\ALG$. 
    % Since the algorithm is non-clairvoyant, the adversary sets the duration of jobs $q^*$ to be $\mu$, and the duration of other jobs to be $1$.  Hence, the cost of $\OPT$ is at most
    %     $2k-1 + \mu$,
    % while the cost of $\ALG$ is at least
    %     $d\mu /2k^2$.

    \[
        \frac{d'\mu /2k^2}{2k-1 + \mu} \ge \frac{d'}{8k^3} \cdot \mu \ge \frac{d'}{\log^3 d'} \cdot \mu,
    \]
    provided
        $\mu \le 2k$,
    and $k \ge 4$. Next, observe that 
    % When $k \ge 4$, we have
    %     $d = {2k \choose k} \cdot k \ge 280$,
    % and
    %     $\mu \le 2k \le \log d$.

    % We have shown the claimed lower bound on competitive ratio on non-clairvoyant algorithms for dimension $d = {2k \choose k} \cdot k$ for any integer $k \geq 4$. For arbitrary dimension $d$, choose $k$ to be the largest for which 
    %     $d \ge d_m = {2k \choose k} \cdot k \ge 2^{2k}$.
    % Then, we can construct an instance for dimension $d$ by simply adding $0$'s in the extra $d-d_m$ dimensions. This proves for dimension $d$ we have a lower bound
    %     $\frac{d_m}{\log^3 d_m} \cdot \mu$.
    % By the choice of $k$, we have
    \begin{align*}
        d &< {2k+1 \choose k+1} \cdot (k+1)
            \le {2k+1 \choose k+1} \cdot k+ 2^{2k+1}
            \le 2 {2k \choose k} \cdot k+ 2^{2k+1}  \\
        &\le 2d' + 2d'  = 4d'
    \end{align*}
    Hence,
        $\frac{d'}{\log^3 d'} \cdot \mu \ge \frac{d}{4 \log^3 d} \cdot \mu$
    which gives the desired bound on the competitive ratio for dimension $d>280$ (which ensures that $k \geq 4$), and 
        $\mu \le \log (d/4) = \log d - 2$.

    Finally, if the model is clairvoyant, the adversary sets all jobs to be of duration $1$, and obtains a lower bound of competitive ratio 
        $(d/k)/(2k) = d/2k^2 = \Omega(d/\log^2 d)$
    that is independent of $\mu$. The other lower bound 
        $\Omega(\sqrt{\log \mu})$
    comes from the lower bound for $1$-dimensional clairvoyant \RSiC established in \cite{Clairvoyant_HA}.
\end{proof}

\subsection{Randomized algorithms}  \label{sec:LB_randomized}

\begin{theorem} \label{thm:LB_randomized_dim1_non-clairvoyant}
    In $1$-dimensional non-clairvoyant \RSiC,
    any randomized algorithm has a competitive ratio at least
        $\frac{1-e^{-1}}{2}\cdot \mu$.
\end{theorem}

\begin{proof}
    We use Yao's principle.  Consider the following distributional input: $k^2$ jobs each of size $1/k$, uniformly at random pick $k$ among $k^2$ jobs to be of duration $\mu$, and let the rest $k^2-k$ jobs to be of duration $1$. 
    Let $\ALG$ be an arbitrary deterministic algorithm. We show that in expectation $\ALG$ has cost $\Omega(k\mu)$. Since $\OPT$ has cost 
        $\le k-1  + \mu$,
    this gives the competitive ratio $\Omega(\mu)$ as desired. 
    
    Let
        $A_1, \ldots, A_m$
    be the $m$ servers that $\ALG$ uses for the above instance. Let
        $|A_i|$
    denote the number of jobs in $A_i$, then
        $1\le |A_i| \le k$.
    We partition  these $m$ servers into $p$ groups 
        $B_1, \ldots, B_{p-1}, B_p$
    such that the number of jobs in each group $B_i$ contains $\ge k$ jobs and $< 2k$ jobs, except perhaps the last group $B_p$ which may contain less than $k$ jobs. Note that such a partition exists by simply partitioning greedily. Hence,
        $p \ge k^2 / 2k = k/2$.
    Let $|B_i|$ denote the number of jobs in group $B_i$.
    Then, for  every 
        $1 \le i \le p-1$, 
    \begin{align*}
        &\Pr[B_i \text{ contains at least one job of duration } \mu]
        = 1 -    \frac{{k^2 -|B_i| \choose k}}{{k^2 \choose k}} \\
        &\ge 1 - \frac{{k^2 -k \choose k}}{{k^2 \choose k}}
        \ge 1 - (1-\frac{1}{k})^k
        \ge 1 - e^{-1}.
    \end{align*}
    Let 
        $X_i \in \{0,1\}$ 
    be a random variable denoting whether group $B_i$ contains some job of duration $\mu$ or not. Then, by the linearity of expectation, the expected cost of $\ALG$ is at least
    \begin{align*}
        \mu \cdot \Ex \left[\sum_{i=1}^{p-1} X_i \right]
        &= \mu \cdot \sum_{i=1}^{p-1} \Ex[X_i] 
        \ge \mu \cdot (p-1) (1-e^{-1})  \\
        &\ge \mu \cdot (k/2 - 1) (1-e^{-1}). 
    \end{align*}
    Hence, the competitive ratio is at least
        $\frac{1-e^{-1}}{2} \cdot \frac{k-2}{k-1+\mu} \cdot \mu$.
    For every $\mu$, since we can pick $k$ to be arbitrarily large, we get the ratio is at least 
        $\frac{1-e^{-1}}{2} \cdot \mu$ 
    as claimed.
\end{proof}

\begin{theorem} \label{thm:LB_randomized_dimd_non-clairvoyant}
    There exists a constant $d_0$ (can take $d_0=280$) such that for every dimension
        $d \ge d_0$ 
    and 
        $\mu \le \log d - 2$,
    every randomized algorithm for $d$-dimensional non-clairvoyant \RSiC has a competitive ratio 
        $\ge \Omega(\frac{d}{ \log^4 d} \cdot \mu)$.
    For $d$-dimensional clairvoyant \RSiC, the lower bound 
        $\Omega(\frac{d}{ \log^2 d})$
    holds.
\end{theorem}

\begin{proof}
    We apply Yao's principle. 
    In \cite{HScoloring}, it is shown that the online graph lower bound construction can be modified to work against randomized algorithms, with only a loss of a constant. We refer the interested reader to the original paper \cite{HScoloring} for  details. Here we only give the necessary modifications, based on the proof given in Theorem \ref{thm:LB_dimd}. The oblivious adversary constructs a distributional input, as follows. 
    \begin{itemize}
        \item The adversary uniformly at random picks a server $Q \in [2k]$; 
        
        \item In step \eqref{eq:choose_At}, the adversary instead chooses a random subset $A_i \subseteq [2k]$ of size $k$;

        \item the vector $\vecs_i$ is defined in the same way;

        \item the adversary assigns $\vecs_i$ to a random server $X \in A_i$;

        \item if $X=Q$, the adversary sets the duration of $\vecs_i$ to be $\mu$, otherwise to be $1$.
    \end{itemize}
     This finishes the construction of the distributional input. Note that the adversary is oblivious. 
     
    By the same argument in \cite{HScoloring}, under this distributional input, one can show that any deterministic algorithm uses in expectation
        $\Omega(d/k)$
    bins. Note that adversary still uses at most $2k$ servers, with cost at most
        $2k-1 + \mu$. 
    Hence, as in the deterministic proof, there exists at least one adversary's server $Q^*$, such that jobs in this server appear in at least
        $\Omega(d/k) / 2k = \Omega(d/k^2)$
    algorithm's bins. Since the adversary randomly picks a server $Q$ and sets jobs in it to be of duration $\mu$, we conclude that in expectation the algorithm has cost at least
        $\ge \frac{1}{2k} \cdot \Omega(d/k^2) \cdot \mu
        = \Omega(d\mu /k^3)$.
    Hence, the expected competitive ratio is at least
        $\frac{\Omega(d\mu /k^3)}{2k-1 + \mu} \ge \Omega(d \mu/k^4)$,
    where
        $k = \Theta(\log d)$,
    as claimed.
    For the clairvoyant case, the lower bound follows by setting all jobs to be of duration $1$.
\end{proof}

\section{Experiments}   \label{sec:Experiment}

% In this section, we investigate the empirical performance of different algorithms for the \RSiC problem on randomly generated inputs. 

 In this section, we provide a thorough evaluation of the average-case performance of almost all existing non-clairvoyant as well as clairvoyant algorithms for the \RSiC problem. 
 %Our study stands out as the most comprehensive one in this field. 
 \subsection{Experimental setup} 

 We evaluate the performance of different algorithms using randomly generated input sequences for $d$-dimensional \RSiC, for $d \in \{1, 2, 5\}$, closely adhering to the experimental setup detailed in \cite{RSiC2015MTF} for the $1$-dimensional case. In the experiments, we assume that each server has size $E^d$ where $E = 1000$, and each job is assumed to have a size in $\{1,2, \cdots, E\}^d$.
 For a given integral span value $T$; for $T \in \{1000, 5000 ,10000\}$, we assume that each job arrives at an integral time step within the interval $[0, T-\mu]$ and has an integral duration in $[1, \mu]$, for $\mu \in \{1,2,5,10,100\}$. Each experimental instance comprises a sequence of $n = 10000$ jobs, with the size and duration of each job selected randomly from their respective ranges, assuming a uniform distribution. The reported upper bound on competitve ratio of each studied algorithm is computed as the ratio of the average cost of the algorithm over 100 input sequences, and the average of the lower bound on $\OPT$ given by Lemma~\ref{lemma:opt_LB} for these instances. 

%Our experiments follow the same set up in \cite{RSiC2015MTF}. 
All our experiments were executed on a personal laptop with a Dual-core 2.3 GHz Intel Core i5 CPU. The laptop had 8 GB of RAM. The laptop was running Mac OS version 12.6.4. The code was written in C++ using VS code version 1.38.1. 

\subsection{Implemented algorithms}
We  implemented both clairvoyant and non-clairvoyant algorithms. The non-clairvoyant algorithms  we implemented are:

\begin{itemize}
    \item \textbf{\NF}:  which keeps only one open server at each time.

    \item \textbf{\textit{Modified}\NF}: which  assigns jobs with sizes greater than a specific threshold separately from the other jobs using the \NF algorithm.

    \item \textbf{\FF}: monotone \AF that orders servers in increasing order of opening time. 
 \item \textbf{\LF}: monotone \AF that orders servers in decreasing order of opening time. 

    \item \textbf{\textit{Modified}\FF}: which assigns jobs with sizes greater than a specific threshold separately from the other jobs using the \FF algorithm.

    \item \textbf{\BF}: \AF  that orders servers in increasing order of 
remaining capacity.

    \item \textbf{\WF}: \AF that orders servers in decreasing order of 
remaining capacity.   

\item \textbf{\RF}: \AF algorithm that orders servers randomly. 

    \item \textbf{\MTF}: monotone \AF  that orders servers in decreasing order of the last time a job was assigned to it. 
 
\end{itemize}

In our experiments, similar to~\cite{RSiC2015MTF}, we adopt the parameters for \textit{Modified}\NF and \textit{Modified}\FF as $E^d/(\mu +1)$ and $E^d/(\mu +7)$, respectively. This choice of values is designed to optimize the competitive ratio of these algorithms, as indicated in \cite{RSiC2015MTF, li2015dynamic}.

The clairvoyant algorithms we implemented are:

\begin{itemize}
    \item \textbf{\textit{Departure Strategy}}~\cite{ren2016clairvoyant}: the span is split into intervals of length $\tau$ each, where $\tau >0$ is a constant. Classifies jobs into categories according to their departure times. Each category contains all  jobs that depart in a time interval of length $\tau$.

    \item \textbf{\textit{Duration Strategy}}~\cite{ren2016clairvoyant}: classifies the jobs into categories such that the max/min job duration ratio for each category is a given constant $\alpha$. Given a base job duration $b$, each category includes all the jobs with durations between $b \alpha ^{i-1}$ and $b \alpha ^{i}$ for an integer $i$.

    \item \textbf{\textit{Hybrid Algorithm} (\HA)}~\cite{Clairvoyant_HA}:  classifies jobs according to their length and their arrivals. Suppose the maximum duration of jobs in the input sequence is $\mu$. Then all the jobs whose lengths are in range $[2^{i-1}, 2^{i}]$ for integer $  1\le i \le \left\lceil \log \mu \right\rceil +1$ and whose arrival times are in time interval $[(c-1) 2^{i}, c2^{i})$ for an integer $c$ are put into the same category. 
  
    \item \textbf{\Greedy}: monotone \AF as defined in Section~\ref{sec:greedy}.

    \item \textbf{New Hybrid}: $\HA^{\oplus d}$ as defined in Section \ref{sec:HA}.

\end{itemize}

\begin{table*}[ht]
\centering
\mysize
\resizebox{\textwidth}{!}{%
\begin{tabular}{|c|c|c|c|c|c|c|c|c|c|c|c|c|c|c|c|}

\hline
\multicolumn{1}{|c|}{} & \multicolumn{5}{|c|}{\textbf{T=1000}} & \multicolumn{5}{|c|}{\textbf{T=5000}} & \multicolumn{5}{|c|}{\textbf{T=10000}} \\
\hline
& $\mu=1$ & $\mu=2$ & $\mu=5$ & $\mu=10$ & $\mu=100$ & $\mu=1$ & $\mu=2$ & $\mu=5$ & $\mu=10$ & $\mu=100$ & $\mu=1$ & $\mu=2$ & $\mu=5$ & $\mu=10$ & $\mu=100$ \\
\hline
\multicolumn{1}{|c|}{} & \multicolumn{15}{|c|}{\textbf{Non-clairvoyant}} \\
\hline
\NF & $1.27$ & $1.37$ & $1.45$ & $1.49$ & $1.52$ & $1.12$ & $1.20$ & $1.32$ & $1.40$ & $ 1.51$ & $1.06$ & $1.10$ & $1.20$ & $1.31$ & $1.50$  \\
\hline
\MNF & $1.31$ & $1.39$ & $1.43$ & $1.48$ & $1.52$ & $1.19$ & $1.29$ & $1.41$ & $1.47$ & $1.52$ & $1.11$ & $1.19$ & $1.31$ & $1.39$ & $1.51$ \\
\hline
\WF & $1.41$ & $1.39$ & $1.36$ & $1.33$ & $1.29$ & $1.16$ & $1.20$ & $1.26$ & $1.28$ & $1.29$ & $1.06$ & $1.09$ & $1.16$ & $1.22$ & $1.29$  \\
\hline
\FF & $1.42$ & $1.36$ & $1.30$ & $1.27$ & $1.22$ & $1.17$ & $1.20$ & $1.24$ & $1.25$ & $1.23$ & $1.07$ & $1.10$ & $1.16$ & $1.21$ & $1.24$  \\
\hline
\MFF & $1.51$ & $1.44$ & $1.35$ & $1.30$ & $1.23$ &$1.25$& $1.29$ & $1.33$ & $1.32$& $1.25$ & $1.13$ & $1.17$ & $1.24$ & $1.28$ & $1.25$  \\
\hline
\BF & $1.51$ & $1.41$ & $1.31$ & $1.24$ & \textbf{1.11} & $1.17$ &$1.21$ & $1.25$ & $1.26$ & \textbf{1.16} & $1.07$ & $1.10$ & $1.17$ & $1.22$ & $1.19$ \\
\hline
\LF & $1.35$ & $1.34$ & $1.29$ & $1.25$ & $1.17$ & $1.14$ & $1.18$ &$1.23$ & $1.24$& $1.19$ & $1.05$ & $1.08$ & $1.15$ & $1.20$ & $1.21$ \\
\hline
\textit{Random Fit} & $1.49$ & $1.41$ & $1.34$ & $1.28$ & $1.18$ & $1.17$& $1.21$ & $1.26$& $1.27$ & $1.21$ & $1.07$ & $1.10$ & $1.17$ & $1.22$ & $1.23$ \\
\hline
\MTF & $1.32$ & $1.32$ & $1.28$ & $1.24$ & $1.16$ & $1.13$ & $1.17$ & $1.22$& $1.24$ & $1.19$ & $1.05$ & $1.08$ & $1.15$ & $1.20$ & $1.20$ \\
\hline
\multicolumn{1}{|c|}{} & \multicolumn{15}{|c|}{\textbf{Clairvoyant}} \\
\hline
\textit{Departure Strategy} & $1.42$ & $1.36$ & $1.30$ & $1.27$ & $1.17$ & $1.17$ & $1.20$ &$1.24$ & $1.25$ & $1.21$ & $1.07$ & $1.10$ & $1.16$ & $1.21$ & $1.23$ \\
\hline
\textit{Duration Strategy} & $1.42$ & $1.40$ & $1.35$ & $1.31$ & $1.23$ & $1.17$ & $1.20$ & $1.24$ & $1.25$ & $1.24$ & $1.07$ & $1.16$ & $1.26$ & $1.33$ & $1.29$  \\
\hline
\textit{Hybrid Algorithm} & \textbf{1.12} & \textbf{1.25} & $1.32$ & $1.33$ & $1.25$ & \textbf{1.03} &  $1.22$& $1.36$ & $1.39$& $ 1.31$& \textbf{1.01} & $1.15$ & $1.30$ & $1.39$ & $1.34$  \\
\hline
\textit{New Hybrid} & \textbf{1.12} & \textbf{1.25} &$1.32$ &$1.33$ & $1.25$ &\textbf{1.03} & $1.22$ & $1.36$ & $1.40$ & $1.31$ & \textbf{1.01} & $1.15$ & $1.30$ & $1.39$ & $1.34$ \\
\hline
\Greedy & $1.28$ & $1.27$ & \textbf{1.22} & \textbf{1.19} & $1.13$ & $1.12$ &\textbf{1.15} & \textbf{1.19} & \textbf{1.20} & \textbf{1.16} & $1.05$ & \textbf{1.07} & \textbf{1.13} & \textbf{1.17} & \textbf{1.17}  \\
\hline
\end{tabular}
}
\caption{Experimental results for the \RSiC problem when $d=1$.}
\label{table: experiment_d=1}
\end{table*}
%%%%%%%%%%%%%%%%%%%%%%%%%%%%%%%%%%%%%%

%%%%%%%%%%%%%%%%%%%%%%%%%%%%%%%%%%%%%%%%
\begin{table*}[ht]
\centering
\mysize
\resizebox{\textwidth}{!}{%
\begin{tabular}{|c|c|c|c|c|c|c|c|c|c|c|c|c|c|c|c|}
\hline
\multicolumn{1}{|c|}{} & \multicolumn{5}{|c|}{\textbf{T=1000}} & \multicolumn{5}{|c|}{\textbf{T=5000}} & \multicolumn{5}{|c|}{\textbf{T=10000}} \\
\hline
& $\mu=1$ & $\mu=2$ & $\mu=5$ & $\mu=10$ & $\mu=100$ & $\mu=1$ & $\mu=2$ & $\mu=5$ & $\mu=10$ & $\mu=100$ & $\mu=1$ & $\mu=2$ & $\mu=5$ & $\mu=10$ & $\mu=100$ \\
\hline
\multicolumn{1}{|c|}{} & \multicolumn{15}{|c|}{\textbf{Non-clairvoyant}} \\
\hline
\NF & $1.40$ & $1.49$ & $1.59$ & $1.65$ & $1.73$ & $1.12$ & $1.20$ & $1.36$& $1.48$& $1.69$ & $1.05$ & $1.09$ & $1.21$ & $1.35$ & $1.65$ \\
\hline
\MNF & $1.44$ & $1.52$ & $1.61$ & $1.65$ & $1.73$  & $1.17$ & $1.25$ & $1.39$ & $1.49$& $1.69$ & $1.09$ & $1.13$ & $1.23$ & $1.36$ & $1.65$ \\
\hline
\WF & $1.46$ & $1.45$ & $1.44$ & $1.42$ & $1.38$  & $1.14$& $1.19$ & $1.29$ & $1.35$ & $1.39$ & $1.05$ & $1.08$ & $1.17$ & $1.26$ & $1.38$ \\
\hline
\FF & $1.49$ & $1.45$ & $1.42$ & $1.40$ & $1.35$ & $1.15$ &$ 1.20$ & $1.29$ & $1.34$ & $1.37$ &  $1.06$ & $1.09$ & $1.17$ & $1.26$ & $1.37$  \\
\hline
\MFF & $1.50$ & $1.47$ & $1.43$ & $1.41$ & $1.35$ & $1.16$ & $1.21$ & $1.30$ & $1.35$ & $1.37$ & $1.06$ & $1.09$ & $1.18$ & $1.26$ & $1.37$ \\
\hline
\BF & $1.48$ & $1.44$ & $1.40$ & $1.37$ & $1.26$ &$1.14$ & $1.19$ & $1.28$ & $1.33$& $1.31$ & $1.05$ & $1.08$ & $1.17$ & $1.25$ & $1.33$  \\
\hline
\LF & $1.39$ & $1.40$ & $1.39$ & $1.36$ & $1.29$ & $1.12$ & $1.17$ & $1.27$ & $1.32$& $1.32$ &  $1.05$ & $1.08$ & $1.16$ & $1.24$ & $1.33$  \\
\hline
\textit{Random Fit} & $1.48$ & $1.45$ & $1.42$ & $1.39$ & $1.30$  & $1.14$& $1.19$ & $1.28$ & $1.34$ & $1.34$ & $1.05$ & $1.08$ & $1.17$ & $1.26$ & $1.35$  \\
\hline
\MTF &$1.38$ & $1.39$ & $1.38$ & $1.36$ & $1.28$ & $1.12$ & $1.17$ & $1.27$ & $1.32$ & $1.32$ &  $1.05$ & $1.07$ & $1.16$ & $1.24$ & $1.33$  \\
\hline
\multicolumn{1}{|c|}{} & \multicolumn{15}{|c|}{\textbf{Clairvoyant}} \\
\hline
\textit{Departure Strategy} & $1.48$ & $1.45$ & $1.42$ & $1.40$ & $ 1.30$  & $1.15$ & $1.20$ & $1.29$ & $1.34$ & $1.35$ & $1.06$ & $1.09$ & $ 1.17$ & $1.26$ & $1.37$ \\
\hline
\textit{Duration Strategy} & $1.48$ & $1.49$ & $1.48$ & $1.47$ & $1.38$  & $1.15$ & $1.20$ & $1.29$ & $1.34$ & $1.37$ & $ 1.06$ & $1.12$ & $1.23$ & $1.34$ & $ 1.44$ \\
\hline
\textit{Hybrid Algorithm} & \textbf{1.23} & $1.37$ & $1.45$ & $1.47$ & $1.38$ & \textbf{1.07} & $1.21$ & $1.36$ & $1.45$ & $1.45$ & \textbf{1.02} & $1.11$ & $1.25$ & $1.37$ & $1.48$ \\
\hline
\textit{New Hybrid} & $1.42$ & $1.54$ &$1.62$ &$1.65$ & $1.64$ & $1.17$ & $1.29$ & $1.46$ & $1.57$ & $1.65$ &  $1.09$ & $1.16$ & $1.31$ & $1.46$ & $1.65$ \\
\hline
\Greedy & $1.36$ & \textbf{1.36} & \textbf{1.34} & \textbf{1.32} & \textbf{1.24} & $1.12$ & \textbf{1.16} & \textbf{1.25} & \textbf{1.30} & \textbf{1.29} & $1.04$ & \textbf{1.07} & \textbf{1.15} & \textbf{1.23} & \textbf{1.30} \\
\hline
\end{tabular}
}
\caption{Experimental results for the \RSiC problem when $d=2$.}
\label{table: experiment_d=2}
\end{table*}
%%%%%%%%%%%%%%%%%%%%%%%%%%%%%%%%%%%%%%
%%%%%%%%%%%%%%%%%%%%%%%%%%%%%%%%%%%%%%%%
\begin{table*}[ht]
\centering
\mysize
\resizebox{\textwidth}{!}{%
\begin{tabular}{|c|c|c|c|c|c|c|c|c|c|c|c|c|c|c|c|}
\hline
\multicolumn{1}{|c|}{} & \multicolumn{5}{|c|}{\textbf{T=1000}} & \multicolumn{5}{|c|}{\textbf{T=5000}} & \multicolumn{5}{|c|}{\textbf{T=10000}} \\
\hline
& $\mu=1$ & $\mu=2$ & $\mu=5$ & $\mu=10$ & $\mu=100$ & $\mu=1$ & $\mu=2$ & $\mu=5$ & $\mu=10$ & $\mu=100$ & $\mu=1$ & $\mu=2$ & $\mu=5$ & $\mu=10$ & $\mu=100$ \\
\hline
\multicolumn{1}{|c|}{} & \multicolumn{15}{|c|}{\textbf{Non-clairvoyant}} \\
\hline
\NF & $1.49$ & $1.58$ & $1.70$ & $1.77$ & $1.90$& $1.11$ & $1.19$ & $1.37$ & $1.52$ & $1.82$ & \textbf{1.04} & $1.08$ & $1.19$ & $1.35$ & $1.76$\\
\hline
\MNF & $1.50$ & $1.58$ & $1.70$ & $1.77$ & $1.90$ & $1.12$ & $1.19$ & $1.37$ & $1.52$ & $1.82$ & \textbf{1.04} & $1.08$ & $1.19$ & $1.35$ & $1.76$\\
\hline
\WF & $1.47$ & $1.52$ & $1.59$ & $1.61$ & $1.61$ & $1.11$ & \textbf{1.18} & \textbf{1.33} & $1.45$ & $1.61$ &  \textbf{1.04} & \textbf{1.07} & $1.18$ & \textbf{1.31} & $1.60$\\
\hline
\FF & $1.48$ & $1.53$ & $1.59$ & $1.62$ & $1.62$ & $1.11$ & \textbf{1.18} & $1.34$ & $1.46$ & $1.62$ & \textbf{1.04} & \textbf{1.07} & $1.18$ & \textbf{1.31} & $1.60$  \\
\hline
\MFF & $1.48$ & $1.53$ & $1.59$ & $1.62$ & $1.62$ & $1.11$ & \textbf{1.18} & $1.34$ & $1.46$ & $1.62$ & \textbf{1.04} & \textbf{1.07} & $1.18$ & \textbf{1.31} & $1.60$ \\
\hline
\BF & $1.47$ & $1.52$ & $1.58$ & $1.60$ & $1.57$ & $1.11$ & \textbf{1.18} & \textbf{1.33} & $1.45$ & $1.60$ &  \textbf{1.04} & \textbf{1.07} & $1.18$ & \textbf{1.31} & $1.59$ \\
\hline
\LF & $1.45$ & $1.51$ & $1.57$ & $1.60$ & $1.57$ & $1.11$ & \textbf{1.18} & \textbf{1.33} & $1.45$ &  $1.60$ &  \textbf{1.04} & \textbf{1.07} & $1.18$ & \textbf{1.31} & \textbf{1.58} \\
\hline
\textit{Random Fit} & $1.47$ & $1.52$ & $1.58$ & $1.61$ & $1.59$ & $1.11$ & \textbf{1.18} & \textbf{1.33} & $1.45$ & $1.61$ & \textbf{1.04} & \textbf{1.07} & $1.18$ & \textbf{1.31} & $1.59$ \\
\hline
\MTF &$1.45$ & $1.51$ & $1.57$ & $1.60$ & $1.57$ & $1.11$ & \textbf{1.18} &  \textbf{1.33} & $1.45$ & $1.60$ & \textbf{1.04} & \textbf{1.07} & $1.18$ & \textbf{1.31} & $1.59$\\
\hline

\multicolumn{1}{|c|}{} & \multicolumn{15}{|c|}{\textbf{Clairvoyant}} \\
\hline
\textit{Departure Strategy} & $1.48$ & $1.53$ & $1.59$ & $1.62$ & $1.61$ & $1.11$ & \textbf{1.18} & $1.34$ & $1.46$ & $1.63$ & \textbf{1.04} & \textbf{1.07} & $1.18$ & \textbf{1.31} & $ 1.61$ \\
\hline
\textit{Duration Strategy} & $1.48$ & $1.55$ & $ 1.63$ & $1.67$ & $1.66$ & $1.11$& \textbf{1.18} & $1.34$ & $1.46$ & $1.63$&  \textbf{1.04} & $1.08$ & $1.19$ & $ 1.34$ & $1.65$  \\
\hline
\textit{Hybrid Algorithm} & \textbf{1.42} & $1.53$ & $1.63$ & $1.68$ & $ 1.68$ & \textbf{1.10} & $1.19$ & $1.37$ & $1.51$ & $1.70$ & \textbf{1.04} & $1.08$ & $1.20$ & $1.35$ & $1.68$ \\
\hline
\textit{New Hybrid} & $1.52$ & $1.61$ &$1.72$ &$1.79$ & $1.91$ & $1.13$ & $1.21$ & $1.39$ & $1.55$ & $1.85$ &  $1.05$ & $1.09$ & $1.21$ & $1.37$ & $ 1.79$ \\
\hline
\Greedy & $1.45$ & \textbf{1.50} & \textbf{1.56} & \textbf{1.59} & \textbf{1.55}& $1.11$ & \textbf{1.18} & \textbf{1.33} & \textbf{1.44} & \textbf{1.59} & \textbf{1.04} & \textbf{1.07} & \textbf{1.17} & \textbf{1.31} & \textbf{1.58} \\
\hline
\end{tabular}
}
\caption{Experimental results for the \RSiC problem when $d=5$.}
\label{table: experiment_d=5}
\end{table*}

% Our experimental results clearly show that \MTF has the best average-case performance when compared to all other non-clairvoyant algorithms as reported in ~\cite{RSiC2015MTF, VectorRSiC2023}. In terms of effectiveness, \LF performs similarly to \MTF, followed by \RF and \FF. On the other hand, \NF, \textit{Modified}\FF, and \textit{Modified}\NF show the worst average-case performance. Among all clairvoyant algorithms, for the small $\mu$s the \textit{Hybrid Algorithm} is the best algorithm while for larger values of $\mu$ the \Greedy algorithm has the best average-case performance. The \textit{Duration Strategy}, on the other hand, shows the worst average-case performance for the small $\mu$s and the \textit{New Hybrid} algorithm has the worst performance for the larger $\mu$s.

\subsection{Experimental results}
Our experimental results for $d \in \{1,2,5\}$ are shown in Tables~\ref{table: experiment_d=1},~\ref{table: experiment_d=2} and ~\ref{table: experiment_d=5}, respectively. We validated our results against those in \cite{RSiC2014} for $d=1$ and for the algorithms implemented there. Our results are slightly different as we use a better lower bound to compute the competitive ratio; when using the same lower bound as \cite{RSiC2014}, our results match exactly.

First we note that for all algorithms, the experimentally derived competitive ratio on random inputs is much better than the worst-case bounds derived theoretically. This is not surprising as the worst-case inputs are carefully constructed to beat the given algorithm and are unlikely to occur in practice. Comparing the three tables, we see that the competitive ratio for every algorithm and every value of $T$ and $\mu$ increases with increasing $d$. In general, the competitive ratio also increases with $\mu$, keeping other parameters constant. 

We note that many algorithms have versions that separate servers into separate categories, and assign jobs to servers in a particular category based on their sizes. In general, such modifications of algorithms do not perform better than the original versions on random inputs, even though they have better worst-case competitive ratios. For example, \MFF performs worse than \FF, \MNF performs worse than \NF, and New Hybrid performs worse than \HA in our experiments. We can also see that all clairvoyant algorithms except \Greedy classify jobs and servers into different categories, and do not have good performance. 

An interesting finding is that \Greedy has the best performance in almost all cases, among all clairvoyant and non-clairvoyant algorithms. Clearly, \Greedy being a clairvoyant algorithm uses the information on finishing time of jobs to its advantage. However, the other clairvoyant algorithms generally do not exhibit good performance, with the exception of \HA for the case $\mu=1$. It is important to note that \Greedy is the only monotone \AF algorithm among the clairvoyant algorithms we have implemented. Note that, the other clairvoyant algorithms do not belong to the monotone \AF algorithms category.

Among non-clairvoyant algorithms, the best algorithms are generally \MTF and \LF, which are also both monotone \AF algorithms. Surprisingly, 
% \yaqiao{this phenomenon was presented in Kamali, but in 2023Murhekar there is no such phenomenon, in fact, there it shows the gap between BestFit and MTF becomes larger as $\mu$ increases, which seems to make more sense, isn't it? because intuitively MTF beats BestFit by better alignment, whose effect should become clearer as $\mu$ increases. I think we should be careful here on what message to show to the reader.} \lata{I have edited slightly to address the comment above and the one in the introduction about experimental results.} 
as in \cite{RSiC2015MTF}, in our experiments \BF is one of the better algorithms, especially for higher values of $\mu$, where its performance ratio equals or betters that of \MTF and \LF \footnote{While \BF also has excellent performance in the experiments of \cite{murhekar2023brief}, it does not beat \MTF for any value of $d$ or $\mu$. However, their experimental setup is somewhat different to ours and that in \cite{RSiC2015MTF}.}.  Recall that the worst-case competitive ratio of \BF is unbounded as shown in \cite{RSiC2014}.

 In~\cite{RSiC2015MTF} and \cite{murhekar2023brief}, two key factors, namely {\em alignment} and {\em packing} are identified as contributing to the performance of an algorithm for \RSiC. The first factor is about how effectively jobs are aligned into servers in terms of their durations, while the second factor evaluates how tightly the jobs are packed together in servers. \AF algorithms (except \WF) do well in terms of packing, but do not consider alignment. \NF tries to align jobs but does not do as well with packing. The authors of \cite{RSiC2015MTF} stipulate that by assigning the next job to the server that has the {\em most recently arrived job}, \MTF succeeds in terms of aligning jobs well in terms of time, while it also succeeds in packing since it is an \AF algorithm.

\Greedy is also an \AF algorithm and therefore does well in terms of packing. By assigning to the server which contains the job that will {\em finish the last}, we may say that \Greedy aligns even better than \MTF. 

A final interesting finding is that the difference in performance between the algorithms appears to narrow for $d=5$. Further research is needed to understand this phenomenon, but one reason could be that because the sizes of jobs in different dimensions are chosen independently, it is harder for any algorithm to achieve a good packing, which diminishes the difference between the algorithms.

\section{Conclusion and open problems}\label{sec:conclusion}

In this paper, we studied the $d$-dimensional \RSiC problem. We introduced a sub-family of \AF algorithms called monotone algorithms, and showed that monotone \AF algorithms achieve competitive ratio $\Theta(d \mu)$. We also introduced a natural clairvoyant algorithm \Greedy that is a monotone \AF algorithm. We showed how to lift $1$-dimensional algorithms for \RSiC to work in $d$ dimensions via a general direct-sum theorem.  We also established a $\tilde{\Omega}(d \mu)$ lower bound for the non-clairvoyant setting under the assumption that $\mu \le \log d - 2$. Finally, we conducted experiments that demonstrate that the newly introduced \Greedy algorithm is among the best performing algorithms (both clairvoyant and non-clairvoyant) in the average-case scenario. 
% There are many directions for future work. The one we would be most interested in is to close the gap between upper and lower bounds in Corollary~\ref{cor:greedy}. In particular, we conjecture that $\rho(Greedy) \le d \mu + O(d)$.

The following is a list of open problems that are suggested by this work:

\begin{itemize}
    \item Close the gap between upper and lower bounds in Corollary~\ref{cor:greedy}. In particular, we conjecture that $\rho(Greedy) \le d \mu + O(d)$.

    \item Does there exist a randomized algorithm for $1$-dimensional clairvoyant \RSiC with competitive ratio $O(1)$?

    \item Can one remove the constraint $\mu \le \log d - 2$ in Theorems~\ref{thm:LB_dimd} and \ref{thm:LB_randomized_dimd_non-clairvoyant}?

    \item For $d$-dimensional clairvoyant \RSiC, is $\Omega(d\sqrt{\log \mu})$ a lower bound for \emph{any} algorithm?

    \item Do the results from the experimental section carry over to real-life industrial instances?
   % \item $\Omega(d\mu)$ lower bound for \emph{any} deterministic algorithm for $d$-dimensional non-clairvoyant \RSiC? We have $\Omega(d\mu/\log^3 d)$.
    
   % \item is there a randomized algorithm achieve $O(1)$ competitive ratio for $1$-dimensional clairvoyant \RSiC? There is nothing.

   % \item $\Omega(d\sqrt{\log \mu})$ for \emph{any} deterministic algorithm for $d$-dimensional clairvoyant \RSiC? We have a weak bound 
   %     $\Omega( \max\{\sqrt{\log \mu}, \frac{d}{ \log^2 d}\})$.
\end{itemize}

\newpage
\balance
~
\bibliography{main}{}
\bibliographystyle{abbrv}

\end{document}